\documentclass[a4paper,UKenglish,cleveref,autoref,thm-restate]{lipics-v2021}

\pdfoutput=1 
\hideLIPIcs 
\nolinenumbers 

\usepackage{amssymb}
\usepackage{hyperref}
\usepackage{graphicx}
\usepackage{booktabs}
\usepackage{forest}

\newcommand{\NN}{\mathbb{N}}

\newcommand{\torel}{\mathrel{\to^=}}
\newcommand{\eps}{\epsilon}
\newcommand{\cc}{\textsf{cc}}
\newcommand{\id}{\textsf{id}}

\newcommand{\vals}{\textsf{vals}}

\newcommand{\vdashc}{\mathrel{\vcenter{\hbox{$\vdash$}}}}

\newcommand{\vDashc}{\mathrel{\vcenter{\hbox{$\vDash$}}}}
\newcommand{\nvDashc}{\mathrel{\vcenter{\hbox{$\nvDash$}}}}

\title{Semantic Labelling in Practice}

\author{Dieter Hofbauer}
{ASW Saarland}{message@dieter-hofbauer.de}{https://orcid.org/0000-0003-2094-6074}{}

\author{Johannes Waldmann}
{HTWK Leipzig}{johannes.waldmann@htwk-leipzig.de}{}{}

\authorrunning{D.~Hofbauer and J.~Waldmann}

\ccsdesc{Theory of computation → Equational logic and rewriting;
  Theory of computation → Rewrite systems}

\keywords{termination, string rewriting, term rewriting}

\begin{document}
\maketitle

\begin{abstract}
  Automating semantic labelling for termination proofs is
  a combinatorially hard problem since the number of algebras
  grows prohibitively large even for small domains.   
  We report on experiments with our tools Matchbox and MnM,
  comparing various model-finding strategies: 
  exhaustive enumeration for bounded domain sizes
  within restricted search spaces, and
  semantic context-closure for fixed algebras.
\end{abstract}

\section{Introduction}

Semantic labelling is a transformation technique for rewrite systems,
introduced by
Zantema~\cite{DBLP:journals/fuin/Zantema95,DBLP:journals/jar/Zantema05}. 
Given a rewrite system, the goal is to find a labelled version for which
termination is easier to prove than for the original system.
For the transformation to be correct, it must preserve both termination
and nontermination.
Since the labelled system typically has a larger alphabet, 
interpretations and other termination methods gain more freedom
as they can assign different meanings to the same original symbol,
now distinguished by labels. 
Each semantic labelling is parameterized by an algebra---hence the name---and
is correct if that algebra is a model of the rewrite system.
Algebras with larger domains yield larger alphabets
for the transformed system---as desired---but automating
semantic labelling becomes increasingly difficult, 
due to the combinatorially large number of algebras. 

In this paper, we compare different model-finding strategies: 
exhaustive enumeration for bounded domain sizes within restricted
search spaces, 
and semantic context-closure for fixed algebras. 
We report on experiments with our tools
Matchbox~\cite{DBLP:conf/rta/Waldmann04}%
\footnote{Available at \url{https://git.imn.htwk-leipzig.de/waldmann/pure-matchbox}.}
and MultumNonMulta (MnM)~\cite{Hofbauer2016}%
\footnote{Available at \url{https://hub.docker.com/repositories/dieterhofbauer/multumnonmulta}.}.

Definitions and examples in this paper refer to string rewriting, but all
approaches discussed can be directly applied to term rewriting as well.
All examples in the directories \texttt{SRS\_Standard} and
\texttt{SRS\_Relative} are taken from the Termination Problems
Database\footnote{TPDB, The Termination Problems Database,
  Version~11.5~\url{https://github.com/TermCOMP/TPDB-ARI}.}, 
and we occasionally refer to results of the annual
Termination Competition (TC)\footnote{%
  For a survey
  see~\url{https://termination-portal.org/wiki/Termination_Competition_History}, 
  and~\url{https://termcomp.github.io/} for the result data.
},
where tools are evaluated on benchmark problems.

\medskip
Following the seminal work of Zantema~\cite{DBLP:journals/fuin/Zantema95}, 
numerous extensions and applications of semantic labelling
have been proposed. Among these, we mention 
self-labelling~\cite{DBLP:conf/cade/MiddeldorpOZ96}; 
labelling for rewriting modulo equations~\cite{DBLP:conf/csl/OhsakiMG00}; 
predictive
labelling~\cite{DBLP:conf/rta/HirokawaM06,DBLP:conf/cade/KoprowskiM07},
restricting the model property to usable rules,
and predictive labelling for innermost
termination~\cite{DBLP:journals/entcs/ThiemannM08};
root-labelling~\cite{DBLP:conf/rta/SternagelM08}; 
modularity results and certification for labelling and
unlabelling~\cite{DBLP:conf/rta/SternagelT11};
SAT-encoding of constraints for labelling, among other proof
methods~\cite{BauEndrullisWaldmann2013,BauThiemannWaldmann2014}. 
Various termination tools include implementations of semantic
labelling, including
Torpa~\cite{DBLP:journals/jar/Zantema05}, 
TPA~\cite{DBLP:conf/rta/Koprowski06a},
Teparla,
Jambox,
AProVE~\cite{DBLP:journals/jar/GieslABEFFHOPSS17},
and TcT~\cite{DBLP:conf/tacas/AvanziniMS16}. 

\medskip
We mostly use standard notations from rewriting, formal languages and algebra,
as in~\cite{BookOtto93, Terese03, DBLP:series/eatcs/Wechler92}. 
For a a set $R$ of strict rules (denoted by $\to$) and
a set $S$ of nonstrict rules (denoted by $\torel$),
we say that $R \cup S$ is \emph{terminating}
if $R$ is terminating relative to $S$.
Let $[0, n) = \{0, \dots, n-1\}$ for $n \in \NN$.

\section{Semantic Labelling}\label{sem-lab}

Semantic labelling is parameterized by a fixed algebra
whose domain serves as the set of labels. 
Since labels propagate right-to-left in strings (and bottom-up in terms),
the labelling is not a homomorphic interpretation.
This right-context sensitivity,
together with the increasing alphabet, 
constitutes the main strength of the transformation. 

For an alphabet $\Sigma$, a \emph{$\Sigma$-algebra} $A$
consists of a nonempty domain $D$
and an interpretation function
$A : \Sigma \to D^D$
that is homomorphically extended to $A : \Sigma^* \to D^D$
by $A(\eps) = \id_D$ for the empty string $\eps$
and the identity $\id_D$ on $D$, and 
$A(x \cdot y) = A(x) \circ A(y)$ for $x, y \in \Sigma^*$.%
\footnote{We use $A$ to denote both the algebra 
  and the interpretation function; this should cause no confusion.}
Thus, the interpretation of strings in a given algebra
is a homomorphic interpretation 
into the function space $D^D$, i.e.,
a homomorphism from $(\Sigma^*, {\cdot})$ into $(D^D, {\circ})$,
where $\cdot$ is concatentation of strings,
and $\circ$ is function composition.
Note that, when letters are viewed as unary symbols of a term signature,
this definition coincides with those based on term evaluation.
In this paper all alphabets and domains are finite; 
without loss of generality,
let $D = [0, n)$ for some natural number $n > 0$.
We further identify $A : \Sigma^* \to D^D$
with its uncurried form $A:  \Sigma^* \times D \to D$
and write $A(x, d)$ for $A(x)(d)$ accordingly.

\begin{remark}\label{automaton}
Such an algebra can be viewed as a
complete deterministic automaton with states $D$
and transitions $d \to^a d'$ if $A(a, d) = d'$.
\end{remark}

\begin{definition}
  An algebra is a \emph{model} of a rewrite system
  $R \subseteq \Sigma^* \times \Sigma^*$,
  denoted by $A \vDashc R$, 
  if $A(\ell) = A(r)$ for every rule $\ell \to r$ in $R$.
\end{definition}

\begin{lemma}\label{equiv}
  For $x, y \in \Sigma^*$, if $A \vDashc R$
  and $x \leftrightarrow^*_R y$, then $A(x) = A(y)$.
\end{lemma}

\begin{definition}
\emph{Semantic labelling} for algebra $A$ is the mapping
$\lambda_A : \Sigma^* \times D \to (\Sigma \times D)^*$, 
defined recursively
for $a \in \Sigma$, $x \in \Sigma^*$, $d \in D$ by
$\lambda_A(\eps, d) = \eps$ and
\[ \lambda_A(a \cdot x, d) = (a, A(x, d)) \cdot \lambda_A(x, d) . \]
\end{definition}

As an immediate consequence, 
$\lambda_A(x \cdot y, d) = \lambda_A(x, A(y, d)) \cdot \lambda_A(y, d)$
for $x,y \in \Sigma^*$. 
We say that $a$ is \emph{labelled} by $d$ in $(a, d)$.
For readability, we abbreviate $(a, d)$ as $a_d$ in examples.
\begin{remark}\label{transducer}
  The semantic labelling function $\lambda x . \lambda_A(x, d)$
  can be defined by a transducer with initial state $d$. 
\end{remark}

Define
$\lambda_A(R) = \{ \lambda_A(\ell, d) \to \lambda_A(r, d) \mid
(\ell \to r) \in R,\, d \in D\}$.

\begin{lemma}\label{label}
  For $x, y \in \Sigma^*$ and $d \in D$,
  $x \to_R y$ implies
  $\lambda(x, d) \to_{\lambda_A(R)} \lambda(y, d)$.
\end{lemma}

By \emph{unlabelling}~\cite[Sect.~7.3]{DBLP:journals/jar/Zantema05}
(cf.~\cite{DBLP:conf/rta/SternagelT11}), 
i.e., by the homomorphism $u : (a, d) \mapsto a$,
we arrive at the reverse implication.\footnote{%
  Observe that for an arbitrary homomorphism $h : \Gamma^* \to \Sigma^*$
  and a rewrite system $S$ over alphabet $\Gamma$, 
  termination of $h(S)$ over alphabet $\Sigma$ implies termination of $S$, 
  since $x \to_S y$ implies $h(x) \to_{h(S)} h(y)$.
}
Note that $u(\lambda_A(R)) = R$.

\begin{lemma}\label{unlabel}
  For $x, y \in (\Sigma \times D)^*$, 
  $x \to_{\lambda_A(R)} y$ implies  $u(x) \to_R u(y)$. 
\end{lemma}

The following result is an immediate consequence of Lemmas~\ref{label}
and~\ref{unlabel}.

\begin{theorem}[{\cite[Thm.~4]{DBLP:journals/fuin/Zantema95}}]\label{main-thm}
  $R$ is terminating if and only if $\lambda_A(R)$ is terminating. 
\end{theorem}

\begin{example}[{\cite[Example~1]{DBLP:journals/fuin/Zantema95}}]\label{aa-aba}
  The algebra $A$ with domain $D = \{0,1\}$, 
  $A(a, d) = 1$ and $A(b, d) = 0$ for $d \in D$
  (i.e., $A = \{ a \mapsto \{0 \mapsto 1, 1 \mapsto 1\},\,
                 b \mapsto \{0 \mapsto 0, 1 \mapsto 0\} \}$)
  is a model of $R = \{a a \to a b a\}$ over alphabet $\{a, b\}$, as
  $A(a a, d) = 1 = A(a b a, d)$ for every $d \in D$. 
  Termination of
  $\lambda_A(R) = \{ a_1 a_0 \to a_0 b_1 a_0 ,\, a_1 a_1 \to a_0 b_1 a_1 \}$
  is easily shown by the weight function\footnote{A weight
    function $w$ is a homomorphic interpretation from
    $(\Sigma^*, {\cdot})$ into $(\NN, {+})$.
    It defines an ordering on $\Sigma^*$ by $x >_w y$ if $w(x) > w(y)$.
    For a rewrite system $R$,
    $\ell >_w r$ for all strict rules and
    $\ell \geq_w r$ for all nonstrict rules $\ell \to r$ in $R$
    proves termination (and linear derivational complexity)
    of $R$.\label{weights}
  }
  $w$ with $w(a_1) = 1$ and $w(c) = 0$ for $c \neq a_1$, 
  whereas termination of $R$ is (slightly) more difficult to prove
  since $R$ is not simply terminating.
\end{example}

\begin{example}
  \texttt{SRS\_Standard/Waldmann\_19/random-373}: 
  Termination of the rewrite system
  $\{b b a a \to a a a b, a a a b \to b a b a,\, a b a b \to a a b b \}$
  was solved by MnM in neither TC'24 nor TC'25
  (but during TC'25 by Matchbox in 29 s and by MUTERM in 1.5 s). 
  After labelling with the model 
  $a \mapsto \{0 \mapsto 1, 1 \mapsto 0\},\,
  b \mapsto \{0 \mapsto 1, 1 \mapsto 0\}$,
  a proof is found by MnM in 0.4 s by 
  a weight ordering and a subsequent 
  matrix interpretation of dimension 5.
\end{example}

\begin{example}\label{saturation-counting} 
  \texttt{SRS\_Relative/Mixed\_relative\_SRS/un03}:
  In the Termination Competition, 
  $R = \{ b a b a b \to b a a b a a a b ,\,
  b a a b a a b \to (b a a a)^3 b ,\,
  (b a a a)^4 b \to b b a a b ,\,
  b b b \torel b a b a a b ,\,\allowbreak 
  b a b a a b \torel b b b \}$
  was shown to be terminating only by
  Matchbox and MnM (TC'25: Matchbox 47 s, MnM 23 s).
  The algebra $A$ with domain $D = [0, 4)$, 
  $A(a, d) = \min\{ d + 1, 3 \}$, and $A(b, d) = 0$
  is a model of $R$, since the first letter
  of each left-hand side is $b$.
  The labelled system has 12 strict and 8 nonstrict rules
  over an alphabet of size 7 (the letter $a_3$ does not appear).
  All strict rules of the labelled system can be removed by some
  weight order, see Footnote~\ref{weights}.
  This algebra uses
  \emph{saturating increment} with maximum $3$
  as interpretation of $a$,
  and the constant \emph{reset} function $d \mapsto 0$ 
  as interpretation of $b$.
  We discuss this family of models further in Section~\ref{semantic-cc}.
\end{example}

Rewrite systems $R$ and $S$ over the same alphabet $\Sigma$
are called \emph{isomorphic} if there exists a
bijection $f : \Sigma \to \Sigma$ such that $f(R) = S$.\footnote{
  In formal language theory such a bijection is known
  as a letter-to-letter renaming.
}
In particular, any (non)termination proof for a rewrite system
carries over to every isomorphic system.

\begin{proposition}
  For isomorphic algebras $A$ and $B$, 
  $\lambda_A(R)$ and $\lambda_B(R)$ are isomorphic.
\end{proposition}

\begin{remark}\label{useless}
  Isomorphic algebras are not the only source of redundancy.
  If the domain is a singleton, the labelled system is isomorphic
  to the original, since all letters carry the same label.  
  More generally for arbitrary domains,
  if $A$ has a fixed point $d \in D$ in the second argument,
  i.e., $A(a, d) = d$ for all $a \in \Sigma$, 
  then the labelled system contains a subset that is isomorphic
  to the original.
  This motivates further restricting the search space to
  strongly connected algebras, see Section~\ref{SC-algebras}. 
\end{remark}

\begin{example}\label{a-eps}
  Any system containing the rule $a \to \eps$
  has only strongly restricted models, 
  as every model must map letter $a$ to $\id_D$. 
  This motivates the use of \emph{context-closures},
  see Example~\ref{a-eps-contd}. 
\end{example}

\section{Restricting the Search Space}

Finding a model for a given rewrite system by enumerating
algebras is feasible only if the search space is sufficiently small,
and
if the subsequent proof attempt for each algebra is sufficiently fast.
In this section, we describe two approaches for reducing the search space:
the use of restricted growth algebras as an approximation of a
full isomorphism check,
and the use of strongly connected algebras to eliminate
a different source of redundancy.
In experiments, we consider only weight orderings
as subsequent proof attempts (see Footnote~\ref{weights}), 
since the existence of a compatible weight ordering can 
be checked efficiently using current constraint solvers;
our tools use GLPK\footnote{GLPK, the GNU Linear Programming Kit,
  \url{https://www.gnu.org/software/glpk/}.} for that purpose.

\subsection{Restricted Growth Algebras}\label{RG-algebras}

Restricted growth sequences are a well-established combinatorial concept
and have been used in various variants for enumerating and counting
discrete structures; see, e.g., \cite[Sect.~7.2.1.5]{Knuth2011}
and~\cite[Sect.~5, Sect.~17.3.4]{Arndt2010}.

\medskip
Without loss of generality, let $\Sigma = [0, s)$ for some $s \in \NN$. 
For algebra $A$ with domain $D = [0, n)$
let $A : a \mapsto \{0 \mapsto a_0, 1 \mapsto a_1, ..., n-1 \mapsto a_{n-1}\}$, 
where $a \in \Sigma$ and $a_i \in D$. 
Let $\vals(a) = (a_0, \dots, a_{n-1})$ be the list of values of $A(a)$, 
and let $(k_1, \dots, k_{s \cdot n})$ be the concatenation
of $\vals(0), \dots, \vals(s-1)$.
Define Algebra $A$ to be a
\emph{restricted growth algebra}, or \emph{RG-algebra} for short,
if $(k_1, \dots, k_{s \cdot n})$  is a restricted growth sequence, 
i.e., if $k_1 \leq 1$ and
$k_i \leq \max\{ k_0, \dots, k_{i-1}\} + 1$
for $1 < i \leq s \cdot n$.
Note that in contrast to the standard definition, 
we also allow $k_1 = 1$ in a
restricted growth sequence.\footnote{%
  Requiring $k_1 = 0$ would force the first letter to map the first
  domain element to itself. This would exclude the case where the first
  domain element is mapped to a different element.\label{RG-0-1}}

\begin{example}\label{expl-RG}
  For $\Sigma = [0,2)$ and $D = [0, 3)$ define $A$ by
  $0 \mapsto \{0 \mapsto 0, 1 \mapsto 0, 2 \mapsto 0\},\,
  1 \mapsto \{0 \mapsto 2, 1 \mapsto 1, 2 \mapsto 1\}$. 
  Here, the concatenation $(k_1, \dots, k_6) = (0, 0, 0, 2, 1, 1)$
  of the two lists 
  $\vals(0) = (0, 0, 0)$ and $\vals(1) = (2, 1, 1)$,
  is not a restricted growth sequence, 
  since $k_4 = 2 > \max\{ k_1, k_2, k_3 \} + 1 = 1$. 
  The isomorphic RG-algebra
  $0 \mapsto \{0 \mapsto 0, 1 \mapsto 0, 2 \mapsto 0\},\,
  1 \mapsto \{0 \mapsto 1, 1 \mapsto 2, 2 \mapsto 2\}$
  is obtained by the letter renaming
  $0 \mapsto 0, 1 \mapsto 2, 2 \mapsto 1$. 
\end{example}

\begin{proposition}
  For every algebra there is an isomorphic RG-algebra. 
\end{proposition}

An algebra $A$ is a model of a rewrite system $R$
if and only if every algebra isomorphic to $A$ is a model of $R$.
Hence, without loss of generality, we may restrict our search for
models of $R$ to RG-algebras.
Note that RG-algebras are not necessarily unique in their respective
isomorphism class, as the next example shows.

\begin{example}\label{RG-z001}
  The two RG-algebras
  $a \mapsto \{0 \mapsto 1, 1 \mapsto 0\},\, b \mapsto \{0 \mapsto 0, 1 \mapsto 0\}$
  and
  $a \mapsto \{0 \mapsto 1, 1 \mapsto 0\},\, b \mapsto \{0 \mapsto 1, 1 \mapsto 1\}$
  are isomorphic via $0 \mapsto 1,\, 1 \mapsto 0$. 
  The RG-algebra
  $a \mapsto \{0 \mapsto 1, 1 \mapsto 0\},\, b \mapsto \{0 \mapsto 1, 1 \mapsto 0\}$,
  however, belongs to a distinct isomorphism class.
  All three algebras are models of $\{ a a b b \to b b b a a a \}$
  (\texttt{SRS\_Standard/Zantema\_04/z001}). 
\end{example}

\subsection{Strongly Connected Algebras}\label{SC-algebras}

A different way to restrict the search space is by
excluding algebras that are not \emph{strongly connected},
cf. Remark~\ref{useless}. 
For $d, d' \in D$, we say that $d'$ is \emph{reachable} from
$d \in D$ if $A(x, d) = d'$ for some $x \in \Sigma^*$.
This corresponds to reachability in the automaton from Remark~\ref{automaton}.
An algebra $A$ is said to be \emph{strongly connected},
or an \emph{SC-algebra} for short, if every element of $D$
is reachable from every other element of $D$.
In particular, every algebra isomorphic to a strongly connected
algebra is strongly connected.
We use the term \emph{strongly connected component (SCC)} in the
standard graph-theoretic sense. 
An SCC is called a \emph{sink} if it has no outgoing edges
in the corresponding condensation DAG. 

\begin{example}\label{expl-SC}
  The algebra $A$ from Example~\ref{expl-RG} is strongly connected.
  In contrast, the algebra
  $B : 0 \mapsto \{0 \mapsto 0, 1 \mapsto 2, 2 \mapsto 0\},\,
  1 \mapsto \{0 \mapsto 0, 1 \mapsto 0, 2 \mapsto 0\}$
  is not strongly connected, since element $1$ is reachable
  from neither $0$ nor $2$.
  Algebra $B$ has three singleton SCCs, and $\{0\}$ is the only sink. 
\end{example}

We generalize the labelling of rules to subsets of the domain
$\bar{D} \subseteq D$ by 
$\lambda_A(R, \bar{D}) = \{ \lambda_A(\ell, d) \to \lambda_A(r, d) \mid
  (\ell \to r) \in R,\, d \in \bar{D}\}$, 
  thus $\lambda_A(R) = \lambda_A(R, D)$.

\begin{proposition}
  For an algebra $A$ with domain $D$, 
  a sink SCC $\bar{D} \subseteq D$ in $A$, 
  and a rewrite system $R$,
  the following are equivalent:
  \begin{enumerate}
  \item\label{one}
    $R$ is terminating. 
  \item\label{two}
    $\lambda_A(R)$ is terminating.     
  \item\label{three}
    $\lambda_A(R, \bar{D})$ is terminating.     
  \end{enumerate}
\end{proposition}

\begin{proof}
  Equivalence of~(\ref{one}) and~(\ref{two}) is Theorem~\ref{main-thm}, 
  and (\ref{two}) implies (\ref{three})
  since $\lambda_A(R, \bar{D}) \subseteq \lambda_A(R)$.
  We show that (\ref{three}) implies (\ref{one}) by contradiction.
  Assume $x_0 \to_R x_1 \to_R \cdots$ is an infinite derivation. 
  The key observation is that 
  for $x \in \Sigma^*$ and any $\bar{d} \in \bar{D}$, 
  $\lambda_A(x, \bar{d}) \in (\Sigma \times \bar{D})^*$
  by the definition of reachability. 
  As 
  $d$ in Lemma~\ref{label} may be chosen arbitrarily,
  we apply the lemma to some 
  $\bar{d} \in \bar{D}$ and
  obtain the infinite derivation 
  $\lambda_A(x, \bar{d}) \to_{\lambda_A(R,\bar{D})}
  \lambda_A(x_1, \bar{d}) \to_{\lambda_A(R,\bar{D})} \cdots$. 
\end{proof}

Again, the search for models can be restricted to SC-algebras, 
and even to algebras that are both RG- and SC-algebras
(RG-SC-algebras for short), 
assuming that smaller domains are considered before larger ones
in the enumeration. 

\begin{example}[Example~\ref{RG-z001} cont'd]
  The three models in Example~\ref{RG-z001} are the only SC-models
  as well as the only SC-RG-models of
  \texttt{SRS\_Standard/Zantema\_04/z001}.
  This system has exactly 10 models with domain size 2. 
\end{example}

\begin{example}\label{SC-z101}
  For $R = \{ a a \to b b b ,\, b b b b b \to a a a \}$
  (\texttt{SRS\_Standard/Zantema\_04/z101}), 
  there is provably no SC-model of domain size $> 1$. 
  Let $A$ be a model of $R$. We have
  $a^{10} \to^*_R b^{15} \to^*_R a^{9}$ and
  $b^{10} \to^*_R a^{6} \to^*_R b^{9}$,
  so $A(a^{9}) = A(a^{10})$ and $A(b^{9}) = A(a^{6}) = A(b^{10})$
  by Lemma~\ref{equiv}. 
  Fix some $d_0 \in D$ and define
  $d = A(a^{9}, d_0)$, $d' = A(a^{10}, d_0)$;
  note that $A(a, d) = d'$. 
  Then $d = A(a^{9}, d_0) = A(a^{10}, d_0) = d'$, 
  thus $A(a, d) = d$ and $A(a^n, d) = d$ for $n \geq 0$.
  From $A(b^{9}, d) = A(b^{10}, d) = A(a^{6}, d) = d$
  we get $A(b, d) = d$.
  This shows that $\{d\}$ is a singleton SCC in $A$.
  Note that $R$ has a termination proof
  by some weight ordering (by design).
\end{example}

\begin{example}
  \texttt{SRS\_Standard/Mixed\_SRS/08}
  (renamed from \texttt{SRS/Endrullis/08}, 
  cf.~\cite{Geser2014}), is
  $E = \{ a a b b \to b b b a ,\, b a \to a a a a \}$. 
  This system has numerous models
  (1 for $|D| = 1$, 6 for $|D| = 2$,
  73 for $|D| = 3$, 1490 for $|D| = 4$, etc.),
  but only one SC-model of dom size 2. 
  Similar to Example~\ref{SC-z101}, we show that
  $E$ has no SC-model of domain size > 2.
  Let $A$ be a model of $E$. We have
  $a^2b^2a \to_E b^3aa \to^*_E a^{11}$ and
  $a^2b^2a \to_E a^2ba^4 \to_E a^{9}$, 
  so $A(a^{9}) = A(a^{11})$ by Lemma~\ref{equiv}. 
  Fix some $d_0 \in D$ and define
  $d = A(a^{9}, d_0)$, $d' = A(a^{11}, d_0)$; 
  note that $A(a^2, d) = d'$. 
  Then $d = A(a^{9}, d_0) = A(a^{11}, d_0) = d'$,
  thus $A(a^2, d) = d$.
  Observe that $a^4 \leftarrow_E ba$,
  thus $A(a^4) = A(ba)$ by Lemma~\ref{equiv}. 
  Case~1: $A(a, d) = d$.
  Then $d = A(a^4, d) = A(ba, d) = A(b, A(a, d)) = A(b, d)$. 
  Case~2: $A(a, d) = \bar{d}$ for some $\bar{d} \neq d$.
  Note that this implies
  $d = A(a^2, d) = A(a, A(a, d)) = A(a, \bar{d})$.
  Then $d = A(a^4, d) = A(ba, d) = A(b, \bar{d})$
  and $\bar{d} = A(a^4, \bar{d}) = A(ba, \bar{d}) = A(b, d)$. 
  Both cases yield an SCC: 
  the singleton $\{d\}$ in Case~1,
  and $\{d, \bar{d}\}$ of size $2$ in Case~2. 
\end{example}

\subsection{Counting Restricted Algebras}

Table~\ref{algebra-counts} compares
the number of $\Sigma$-algebras with domain size $n$
against the cardinality of the restricted classes discussed above.
The number of all algebras is $n^{n |\Sigma|}$,
the number of all RG-algebras can explicitely be given
using Stirling numbers of the second kind,
and for the number of all SC-algebras see 
Table~1 in~\cite{DBLP:journals/iandc/Radke65}.

\begin{table}[h]
  \centering
    \caption{Number of algebras with domain size $n$ and alphabet $\Sigma$}
  \label{algebra-counts}
  \begin{tabular}{l *{5}{r}}
    \toprule
    $n$ & $|\Sigma|$
    & all & RG & SC & RG-SC \\[0cm]
    \midrule
    1 & $\ast$ & 1 & 1 & 1 & 1 \\ 
    \hline
    2 & 1 & 4 & 4 & 1 & 1 \\
    2 & 2 & 16 & 16 & 9 & 9 \\
    2 & 3 & 64 & 64 & 49 & 49 \\
    2 & 4 & 256 & 256 & 225 & 225 \\
    2 & 5 & 1024 & 1024 & 961 & 961 \\
    2 & 6 & 4096 & 4096 & 3969 & 3969 \\
    2 & 7 & 16384 & 16384 & 16129 & 16129 \\
    \hline
    3 & 1 & 27 & 14 & 2 & 1 \\
    3 & 2 & 729 & 365 & 296 & 144 \\
    3 & 3 & 19683 & 9842 & 13754 & 6759 \\
    3 & 4 & 531441 & 265721 & 458000 & 227118 \\
    3 & 5 & 14348907 & 7174454 & 13474802 & 6712575 \\
    \hline    
    4 & 1 & 256 & 51 & 6 & 1 \\
    4 & 2 & 65536 & 11051 & 20958 & 3286 \\
    4 & 3 & 16777216 & 2798251 & 11127270 & 1792394 \\
    \hline
    5 & 1 & 3125 & 202 & 24 & 1 \\
    5 & 2 & 9765625 & 422005 & 2554344 & 96825  \\
    \hline
    6 & 1 & 46656 & 876 & 120 & 1 \\
    6 & 2 & 2176782336 & 19628064 & 474099840 & 3483414 \\
    \bottomrule 
  \end{tabular}
\end{table}

\section{Context-Closure}

Example~\ref{a-eps} illustrates that the class of models
for systems with projection rules
(i.e., rules with empty right-hand sides)
is strongly restricted. 
Similar restrictions frequently arise in practice
for a variety of systems.
In such cases, an appropriate \emph{contex-closure}
may turn the system into one that admits a model, 
enabling subsequent proof steps to succeed. 

We use the following notations for concatenation on languages,
rewrite rules and systems.
For $x \in \Sigma$ and $X, Y \subseteq \Sigma^*$ let
$xY = \{xy \mid y \in Y\}$ and $XY = \bigcup_{x \in X}xY$. 
For a rewrite rule $\ell \to r$ 
let $x(\ell \to r) = (x\ell \to xr)$ be its \emph{context-closure}
with \emph{left-context} $x$, 
further let
$X(\ell \to r) = \bigcup_{x \in X}(x\ell \to xr)$,
and the system $XR = \bigcup_{p \in R}Xp$
is the \emph{context-closure} of $R$ with \emph{left-contexts} $X$; 
closure under right-contexts is defined analogously.
For $n, m \in \NN$, let $\cc(n, R, m)$ denote
$\Sigma^n R \Sigma^m$, abbreviated as $\cc(n, m)$
when $R$ is clear from the context. 

\begin{definition}
  For rewrite rules $p'$ and $p$ define $p'$ to be a
  \emph{context-closure} of $p$, if $p' = xpy$ for some $x, y \in \Sigma^*$.   
  For rewrite systems $R'$ and $R$ define
  $R'$ to be a \emph{context-closure} of $R$, if
  each rule in $R'$ is the context-closure of some rule in $R$. 
\end{definition}

For an arbitrary context-closure $R'$ of $R$, termination of
$R$ implies termination of $R'$,
since ${\to_{R'}} \subseteq {\to_R}$, 
but the converse does not hold in general. 

\begin{example}\label{a-ba}
  Context-closure can turn a nontermining system
  into a terminating one, 
  as for $\{a \to b a\}$ and its context-closure
  $a \{a \to b a\} = \{a a \to a b a\}$,
  cf.\ Example~\ref{aa-aba}.
\end{example}

Nontermination is preserved, however, if we additionally
require the context-closure to be \emph{complete}.
For $x, y \in \Sigma^*$ we write $y \geq x$ if $x$ is a suffix of $y$,
i.e., if $y \in \Sigma^*x$.
A set $X \subseteq \Sigma^*$ is a \emph{complete set of suffixes},
if $\Sigma^* = Z \cup \Sigma^* X$
for $Z = \{z \mid \exists x \in X : x \geq z\}$,
i.e, each string is the suffix of some string in $X$,
or has a suffix in $X$.

\begin{example}
  The set $X = \{a, aab, bab, bb\} $ over alphabet $\{a, b\}$ is a complete
  set of suffixes, and it is a minimal such set under set inclusion.%
  \footnote{
  If we additionally require complete sets of suffixes to be minimal
  under set inclusion, after string reversal
  we arrive at the notion of a maximal prefix code,
  cf.~\cite{BerstelPerrinReutenauer2010}.
  }
  Note that $X$ is not a complete set of prefixes (defined
  symmetrically) since $baa$ is neither in $X \Sigma^*$
  nor a prefix of some string in $X$.
\end{example}

In what follows, we restrict ourselves to left-contexts.

\begin{definition}\label{complete-cc}
  A rewrite system $R'$ is a \emph{complete} 
  context-closure of a rewrite system $R$,
  if $R' = \bigcup_{p \in R}X_pp$
  for some complete sets of suffixes $X_p$.
\end{definition}

\begin{proposition}
  Any complete finite context-closure of $R$
  is terminating if and only if $R$ is terminating. 
\end{proposition}

\begin{proof}
  Let $R'$ be a finite complete context-closure of $R$
  as in Definition~\ref{complete-cc}, 
  and let $m = \max\{ |x| \mid \exists p \in R, x \in X_p \}$.
  Then for $z \in \Sigma^m$, 
  $x \to_R y$ implies $z x \to_{R'} z y$, 
  so any infinite derivation $x_0 \to_R x_1 \to_R \dots$
  carries over to the infinite derivation
  $z x_0 \to_R z x_1 \to_R \dots$.   
\end{proof}

\begin{remark}\label{expand}
  If $X \cup \{x\}$ is a complete set of suffixes, 
  then $X \cup \{\Sigma x\}$ is a again complete set of suffixes.
  Note that any finite complete set of suffixes can be
  obtained by this expansion process, starting with $\{\eps\}$. 
  More generally, if both $X \cup \{x\}$ and $Y$ are complete sets of suffixes, 
  then $X \cup \{Y x\}$ is a complete set of suffixes.
\end{remark}

\begin{remark}\label{preserve-model}
  An arbitrary context-closure preserves the model property:
  If $A \vDashc (\ell \to r)$, then
  $A \vDashc x(\ell \to r)y$,
  since $A(\ell) = A(r)$ implies
  $A(x \ell y) = A(x) \circ A(\ell) \circ A(y) = 
  A(x) \circ A(r) \circ A(y) = A(x r y)$. 
\end{remark}

Conversely, a context-closure may establish the
model property, turning a non-model into a model: 
Even if $A(\ell) \neq A(r)$,
we can have $A(x \ell) = A(x r)$ in case $A(x)$ is not injective.

\begin{example}\label{a-eps-contd}
Let $R = \{a \to \eps , b \to \eps \}$, cf.~Example~\ref{a-eps}. 
For arbitrary domain $D$, this system has just
the one model $a, b \mapsto \id_D$, which is not an SC-model
for domain size $> 1$.
In constrast, the 4-rules system $\{a, b\}R$ 
has exactly 5 models of domain size 2, 
of which 2 are RG-SC-models,
see Sections~\ref{RG-algebras} and~\ref{SC-algebras}. 
\end{example}

\begin{remark}\label{finite}
  If the algebra domain $D$ is finite, then $D^D$ is finite as well. 
  Therefore, although $\Sigma^* x$ is infinite, we can
  effectively compute the set of functions $A(\Sigma^* x)$
  for any $x \in \Sigma^*$
  by expanding contexts as in Remark~\ref{expand},
  and detect termination of this expansion process: 
  If $A(\Sigma^{i+1} x) = A(\Sigma^i x)$,
  then $A(\Sigma^n x) = A(\Sigma^i x)$ for $n \geq i$.
\end{remark}

\subsection{Full Context-Closure}\label{full-cc}

This approach proceeds by successively searching
for models of increasingly large context-closures of $R$,
i.e., of $\Sigma^i R$ for $i \geq 0$.
Cleary, each $\Sigma^i R$ is a complete context-closure of $R$.
    
\begin{example}\label{W23-10-2-97}
  \texttt{SRS\_Relative/Waldmann-23/size-10-alpha-2-num-97}: 
  This 3-rules system
  $\{ a a \to \eps , b b b \to \eps , a \torel a b b a \}$
  was solved only by Matchbox in TC'25 (34 s) and in TC'24 (50 s),
  and was unsolved before.
  
  We apply full context-closure $\cc(k,0)$ for increasing $k$, 
  and use weight orderings as the only subsequent proof attempt. 
  No termination proof is found for $k \leq 3$,
  but for $k \geq 4$, we succeed.
  Each row in the Tables~\ref{table-cc} and~\ref{table-cc-projections}
  corresponds to a search through
  9765625 (all), 422005 (RG), 2554344 (SC), or 96825 (RG-SC) algebras
  respectively, cf.\ Table~\ref{algebra-counts}.
  Each entry gives the number of models tried until 
  either the search failed or a termination proof was found, 
  along with the corresponding runtime.
  As stated in Remark~\ref{preserve-model},
  if $\cc(k,0)$ succeeds, then also $\cc(k+1,0)$ succeeds, 
  so the columns for $\cc(5,0)$ and $\cc(6,0)$ are only included
  to display the corresponding runtime.
  In Table~\ref{table-cc}, full context-closure is applied to all rules,
  and in Table~\ref{table-cc-projections} only to projection rules. 

  \medskip
  These results suggest the following interpretation.
  Regarding unsuccessful proof attempts, the runtime is generally
  poor---as expected, since exhausting all possibilities is expensive. 
  For successful attempts, SC-algebras perform well,
  with RG-SC-algebras performing even better.
  Furthermore, restricting the context-closure to a subset of the
  system---here, the projection rules---can be beneficial,
  reducing both search space and runtime.
  In Table~\ref{table-cc-projections},
  already the first RG-SC-model encountered in the enumeration succeeds. 
  Finally, larger contexts need not result in prohibitive runtime,
  suggesting that the approach scales reasonably well in practice.
\end{example}

\begin{table}[h]
  \caption{Models for Example~\ref{W23-10-2-97} with $n = 5$
    and context-closure 
    for all rules}
  \label{table-cc}
  \centering
  \begin{tabular}{r|rr|rrr}
    \toprule
    \multicolumn{1}{r|}{}
    & \multicolumn{2}{c|}{unsuccessful}
    & \multicolumn{3}{c}{successful} \\
    \cmidrule(lr){2-3} \cmidrule(lr){4-6}
    models & $\cc(2,0)$ & $\cc(3,0)$ & $\cc(4,0)$ & $\cc(5,0)$ & $\cc(6,0)$ \\
    \midrule
    all & 16076 (77 s) & 71036 (298 s) & 195 (2.0 s) & 195 (3.7 s) & 195 (8.9 s) \\    
    RG & 1219 (4.8 s) & 4325 (16 s) & 52 (0.8 s) & 52 (1.4 s) & 52 (2.8 s) \\
    SC & 0 (26 s) & 960 (32 s) & 1 (0.2 s) & 1 (0.3 s) & 1 (0.3 s) \\
    {\bf RG-SC} & {\bf 0 (1.2 s)} & {\bf 40 (1.6 s)} & {\bf 1 (0.2 s)} & {\bf 1 (0.2 s)} & {\bf 1 (0.3 s)} \\
    \bottomrule 
  \end{tabular}
\end{table}

\begin{table}[h]
  \caption{Models for Example~\ref{W23-10-2-97} with $n = 5$ and 
    context-closure 
    for projection rules only}
  \label{table-cc-projections}
  \centering
  \begin{tabular}{r|rr|rrr}
    \toprule
    \multicolumn{1}{r|}{}
    & \multicolumn{2}{c|}{unsuccessful}
    & \multicolumn{3}{c}{successful} \\
    \cmidrule(lr){2-3} \cmidrule(lr){4-6}
    models & $\cc(2,0)$ & $\cc(3,0)$ & $\cc(4,0)$ & $\cc(5,0)$ & $\cc(6,0)$ \\
    \midrule
    all & 3456 (44 s) & 5436 (61 s) & 195 (1.5 s) & 195 (2.5 s) & 195 (5.5 s) \\    
    RG & 258 (2.3 s) & 363 (2.8 s) & 52 (0.6 s) & 52 (1.0 s) & 52 (1.9 s) \\    
    SC & 0 (24 s) & 0 (24 s) & 1 (0.2 s) & 1 (0.3 s) & 1 (0.3 s) \\
    {\bf RG-SC} & {\bf 0 (1.2 s)} & {\bf 0 (1.2 s)} & {\bf 1 (0.2 s)} & {\bf 1 (0.2 s)} & {\bf 1 (0.3 s)} \\    
    \bottomrule 
  \end{tabular}
\end{table}

\subsection{Semantic Context-Closure}\label{semantic-cc}

Here, we fix an algebra $A$
and test whether it is a model
of increasingly large context-closures of $R$.
Completeness of these context-closures is ensured by
expanding context-closures for single rules by
single letters in all possible ways.
This is done for those rules that violate the model property.

More formally, 
define the relation $\vdashc_A$ on rewrite systems by
$R \cup \{p\} \vdashc_A R \cup \Sigma p$,
where $p$ is a 
rule with $A \nvDashc p$.
We start with $R_0 = R$ and 
compute a sequence $R_0 \vdashc_A R_1 \vdashc_A R_2 \cdots$.
This process need not terminate, 
but if it does, the result is uniquely determined, 
as the relation $\vdashc_A$ is confluent. 
Further, the corresponding set of suffixes
is complete, because $\{\eps\}$ is complete
(note that $R = \{\eps\}R$) 
and completeness is preserved by $\vdashc_A$
(cf.~Remark~\ref{expand}).

\begin{remark}
  Termination of $\vdashc_A$ is decidable,
  exploiting the finiteness of $D^D$ (cf.~Remark~\ref{finite}). 
  For each rule $p$ define a finitely branching tree
  of rewrite rules with root $p$.
  A rule $p'$ is a leaf in case $A \vDashc p'$,
  otherwise its children are the rules in $\Sigma p'$. 
  To each node $\ell \to r$ associate the pair
  $(A(\ell), A(r)) \in D^D \times D^D$,
  thus $\ell \to r$ is a leaf iff the pair
  satisfies $A(\ell) = A(r)$.
  By König's lemma, $\vdashc_A$ is nonterminating
  iff some tree has an infinite path. 
  Suppose a path contains nodes
  $x(\ell \to r)$ and $yx(\ell \to r)$, 
  for $x \in \Sigma^*$ and $y \in \Sigma^+$,
  with identical pairs associated,
  and $A(x\ell) \neq A(xr)$.
  Since $A$ is a homomorphism,
  $A(y^k x \ell) = A(x\ell) \neq A(xr) = A(y^k x r)$
  for $k \ge 0$.
  Therefore, no node $y^k x(\ell \to r)$ is a leaf,
  and the tree has an infinite path. 
  Since $D^D \times D^D$ is finite, any path of length greater than
  $|D^D \times D^D|$ contains two nodes with identical pairs. 
  Expanding each tree to depth $|D^D \times D^D| + 1$
  therefore decides termination.
\end{remark}

\begin{remark}
  If the process terminates, the result is minimal
  with respect to $\leq_\cc$, 
  where 
  $R \leq_\cc R'$ if $R'$ is a context-closure of $R$.
  Adding further contexts to such a minimal set is useless,
  because the model property is preserved under context-closures,
  see Remark~\ref{preserve-model}.
  By definition, if the process terminates, the result is a
  model of $A$.
  Thus, whenever a 
  context-closure of $R$ exists that is a model of $A$,
  then such a system can be effectively computed by this approach. 
\end{remark}

Finally, we present a class of algebras for which 
semantic context closure is guaranteed to terminate. 
For $n > 0$ and nonempty $\Gamma \subseteq \Sigma$ let
$I_{n,\Gamma}$ be the $\Sigma$-algebra with domain $[0,n)$, 
$I_{n,\Gamma}(a,d) = \min\{d+1, n-1\}$ for $a \in \Gamma$
(\emph{saturating increment}), and
$I_{n,\Gamma}(a,d) = 0$
otherwise (\emph{reset}).
Note that $I_{n,\Gamma}$ is a strongly connected restricted growth algebra.

\begin{example}
  The algebra $I_{4,\{a\}}$was used in Example~\ref{saturation-counting}. 
\end{example}

\begin{example}[Example~\ref{W23-10-2-97} cont'd]\label{W23-10-2-97-contd}
  Consider again
  $R = \{ a a \to \eps , b b b \to \eps , a \torel a b b a \}$
  over $\Sigma = \{a,b\}$,
  and choose $A = I_{5,\Gamma}$ for $\Gamma = \{b\}$. 
  Algebra $A$ already is a model of the third rule, since both its
  left- and right-hand side start with letter $a$. 
  The expansion process thus adds contexts for the two
  remaining (projection) rules.
  For instance, 
  $R_0 = R \setminus \{a \torel a b b a \} \vdashc_A
  R_1 = R_0 \setminus \{aa \to \eps\} \cup \{aaa \to a,\, baa \to b\}
  \vdashc_A
  R_1 \setminus \{baa \to b\} \cup \{abaa \to ab,\, bbaa \to bb\}
  \vdashc_A \cdots \vdashc_A S R_0$
  for the resulting complete set of suffixes 
  $S = \{ a b^n \mid 0 \leq n \leq 3 \} \cup \{ b^4 \}$. 
  The labelled system has 50 strict and 5 nonstrict rules
  over an alphabet of size 10.
  A weight ordering finally establishes termination. 
  Figure~\ref{expansion-tree} illustates the expansion of suffixes. 
  In this example, for both projection rules 
  the resulting sets of suffixes happen to coincide.

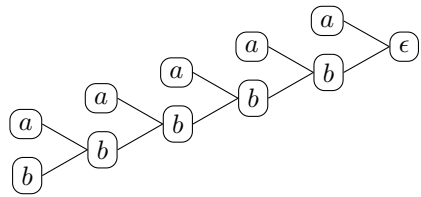
\begin{figure}
\centering
\begin{forest}
  for tree={draw, rounded corners, grow=west,
    anchor=east, parent anchor=west, child anchor=east}
  [$\eps$
  [$a$]
    [$b$
      [$a$]
      [$b$
        [$a$]
        [$b$
          [$a$]
          [$b$
            [$a$]
            [$b$]
          ]
        ]
      ]
    ]
  ]
\end{forest}
\caption{Suffix expansion for Example~\ref{W23-10-2-97-contd}}
\label{expansion-tree}
\end{figure}

\end{example}

For an arbitrary rule $\ell \to r$,
semantic context closure with $A = I_{n,\Gamma}$ always terminates: 
We have $A \vDashc x(\ell \to r)$ for every $x \in \Sigma^{n-1}$ ($\ast$),
therefore semantic context-closure stops with some subset of
suffixes of $\Sigma^{n-1}$. 
In order to prove ($\ast$), let $\bar{\Gamma} = \Sigma \setminus \Gamma$
and consider two cases:
If $x = yz$ for some $z \in \bar{\Gamma}\Sigma^*$, then
$A \vDashc z(\ell \to r)$, since $A(z\ell, d) = 0 = A(zr, d)$ for $d \in D$,
thus $A \vDashc yz(\ell \to r)$ by Remark~\ref{preserve-model}.
Otherwise we have $x \in \Gamma^*$,
and $A \vDashc x(\ell \to r)$,
since $A(x \ell, d) = n - 1 = A(x r, d)$ for $d \in D$. 
Note that this proof also shows that the resulting suffix set is
always a subset of $\bar{\Gamma} \Gamma^* \cup \Gamma^*$.

\section{Future Work}

All constructions presented in this paper extend to
quasi-models~\cite{DBLP:journals/fuin/Zantema95}.
Continuing the approach from Section~\ref{semantic-cc}, 
a natural direction is to investigate
semantic context-closure for other families of algebras,
or to show that the family of algebras considered here
is in some sense canonical. 
Finally, it may be worth exploring the relation between
finite models and 
interpretations with almost linear weight functions
over an infinite domain~\cite{Hofbauer-18b}.

\bibliography{main.bib}

\begin{thebibliography}{10}

\bibitem{Arndt2010}
Jörg Arndt.
\newblock {\em Matters Computational: Ideas, Algorithms, Source Code}.
\newblock 2010.
\newblock URL: \url{http://www.jjj.de/fxt/#fxtbook}.

\bibitem{DBLP:conf/tacas/AvanziniMS16}
Martin Avanzini, Georg Moser, and Michael Schaper.
\newblock {TcT}: Tyrolean complexity tool.
\newblock In Marsha Chechik and Jean{-}Fran{\c{c}}ois Raskin, editors, {\em
  Tools and Algorithms for the Construction and Analysis of Systems, 22nd
  International Conference, {TACAS} 2016, Held as Part of the European Joint
  Conferences on Theory and Practice of Software, {ETAPS} 2016, Eindhoven, The
  Netherlands, April 2-8, 2016, Proceedings}, Lecture Notes in Computer
  Science, pages 407--423. Springer, 2016.
\newblock \href {https://doi.org/10.1007/978-3-662-49674-9_24}
  {\path{doi:10.1007/978-3-662-49674-9_24}}.

\bibitem{BauEndrullisWaldmann2013}
Alexander Bau, Jörg Endrullis, and Johannes Waldmann.
\newblock {SAT} compilation for termination proofs via semantic labelling.
\newblock In Johannes Waldmann, editor, {\em 13th International Workshop on
  Termination, WST 2013, Bertinoro, Italy, August 29-31, 2013. Proceedings},
  pages 8--12, 2013.
\newblock URL: \url{https://termination-portal.org/wiki/WST}.

\bibitem{BauThiemannWaldmann2014}
Alexander Bau, René Thiemann, and Johannes Waldmann.
\newblock Automated {SAT} encoding for termination proofs with semantic
  labelling and unlabelling.
\newblock In Carsten Fuhs, editor, {\em 14th International Workshop on
  Termination, WST 2014, Vienna, Austria, July 17-18, 2014. Proceedings}, pages
  6--10, 2014.
\newblock URL: \url{https://termination-portal.org/wiki/WST}.

\bibitem{BerstelPerrinReutenauer2010}
Jean Berstel, Dominique Perrin, and Christophe Reutenauer.
\newblock {\em Codes and Automata}, volume 129 of {\em Encyclopedia of
  mathematics and its applications}.
\newblock Cambridge University Press, 2010.

\bibitem{BookOtto93}
Ronald~V. Book and Friedrich Otto.
\newblock {\em String-Rewriting Systems}.
\newblock Texts and Monographs in Computer Science. Springer, 1993.
\newblock \href {https://doi.org/10.1007/978-1-4613-9771-7}
  {\path{doi:10.1007/978-1-4613-9771-7}}.

\bibitem{Geser2014}
Alfons Geser.
\newblock A solution to endrullis-08 and similar problems.
\newblock In Carsten Fuhs, editor, {\em 14th International Workshop on
  Termination, WST 2014, Vienna, Austria, July 17-18, 2014. Proceedings}, pages
  31--35, 2014.
\newblock URL: \url{https://termination-portal.org/wiki/WST}.

\bibitem{DBLP:journals/jar/GieslABEFFHOPSS17}
J{\"{u}}rgen Giesl, Cornelius Aschermann, Marc Brockschmidt, Fabian Emmes,
  Florian Frohn, Carsten Fuhs, Jera Hensel, Carsten Otto, Martin Pl{\"{u}}cker,
  Peter Schneider{-}Kamp, Thomas Str{\"{o}}der, Stephanie Swiderski, and
  Ren{\'{e}} Thiemann.
\newblock Analyzing program termination and complexity automatically with
  {AProVE}.
\newblock {\em J. Autom. Reason.}, 58(1):3--31, 2017.
\newblock \href {https://doi.org/10.1007/S10817-016-9388-Y}
  {\path{doi:10.1007/S10817-016-9388-Y}}.

\bibitem{DBLP:conf/rta/HirokawaM06}
Nao Hirokawa and Aart Middeldorp.
\newblock Predictive labeling.
\newblock In Frank Pfenning, editor, {\em Term Rewriting and Applications, 17th
  International Conference, {RTA} 2006, Seattle, WA, USA, August 12-14, 2006,
  Proceedings}, Lecture Notes in Computer Science, pages 313--327. Springer,
  2006.
\newblock \href {https://doi.org/10.1007/11805618_24}
  {\path{doi:10.1007/11805618_24}}.

\bibitem{Hofbauer2016}
Dieter Hofbauer.
\newblock {System description: MultumNonMulta}.
\newblock In Aart Middeldorp and Ren{\'{e}} Thiemann, editors, {\em 15th
  International Workshop on Termination, {WST} 2016, Obergurgl, Austria,
  September 5-7, 2016. Proceedings}, page~90, 2016.
\newblock URL: \url{https://termination-portal.org/wiki/WST}.

\bibitem{Hofbauer-18b}
Dieter Hofbauer.
\newblock Embracing infinity -- termination of string rewriting by almost
  linear weight functions.
\newblock In Salvador Lucas, editor, {\em 16th International Workshop on
  Termination, WST 2018, Oxford, U.~K., July 18-19, 2018. Proceedings}, pages
  65--69, 2018.
\newblock URL: \url{https://termination-portal.org/wiki/WST}.

\bibitem{Knuth2011}
Donald~Ervin Knuth.
\newblock {\em The Art of Computer Programming, Volume {4A}: Combinatorial
  Algorithms, Part 1}.
\newblock Addison-Wesley, Upper Saddle River, New Jersey, 2011.

\bibitem{DBLP:conf/rta/Koprowski06a}
Adam Koprowski.
\newblock {TPA:} termination proved automatically.
\newblock In Frank Pfenning, editor, {\em Term Rewriting and Applications, 17th
  International Conference, {RTA} 2006, Seattle, WA, USA, August 12-14, 2006,
  Proceedings}, Lecture Notes in Computer Science, pages 257--266. Springer,
  2006.
\newblock \href {https://doi.org/10.1007/11805618_19}
  {\path{doi:10.1007/11805618_19}}.

\bibitem{DBLP:conf/cade/KoprowskiM07}
Adam Koprowski and Aart Middeldorp.
\newblock Predictive labeling with dependency pairs using {SAT}.
\newblock In Frank Pfenning, editor, {\em Automated Deduction, 21st
  International Conference, CADE-21, Bremen, Germany, July 17-20, 2007,
  Proceedings}, Lecture Notes in Computer Science, pages 410--425. Springer,
  2007.
\newblock \href {https://doi.org/10.1007/978-3-540-73595-3_31}
  {\path{doi:10.1007/978-3-540-73595-3_31}}.

\bibitem{DBLP:conf/cade/MiddeldorpOZ96}
Aart Middeldorp, Hitoshi Ohsaki, and Hans Zantema.
\newblock Transforming termination by self-labelling.
\newblock In Michael~A. McRobbie and John~K. Slaney, editors, {\em Automated
  Deduction, 13th International Conference, CADE-13, New Brunswick, NJ, USA,
  July 30 - August 3, 1996, Proceedings}, Lecture Notes in Computer Science,
  pages 373--387. Springer, 1996.
\newblock \href {https://doi.org/10.1007/3-540-61511-3_101}
  {\path{doi:10.1007/3-540-61511-3_101}}.

\bibitem{DBLP:conf/csl/OhsakiMG00}
Hitoshi Ohsaki, Aart Middeldorp, and J{\"{u}}rgen Giesl.
\newblock Equational termination by semantic labelling.
\newblock In Peter Clote and Helmut Schwichtenberg, editors, {\em Computer
  Science Logic, 14th Annual Conference of the EACSL, Fischbachau, Germany,
  August 21-26, 2000, Proceedings}, Lecture Notes in Computer Science, pages
  457--471. Springer, 2000.
\newblock \href {https://doi.org/10.1007/3-540-44622-2_31}
  {\path{doi:10.1007/3-540-44622-2_31}}.

\bibitem{DBLP:journals/iandc/Radke65}
Charles~E. Radke.
\newblock Enumeration of strongly connected sequential machines.
\newblock {\em Inf. Control.}, 8(4):377--389, 1965.
\newblock \href {https://doi.org/10.1016/S0019-9958(65)90316-5}
  {\path{doi:10.1016/S0019-9958(65)90316-5}}.

\bibitem{DBLP:conf/rta/SternagelM08}
Christian Sternagel and Aart Middeldorp.
\newblock Root-labeling.
\newblock In Andrei Voronkov, editor, {\em Rewriting Techniques and
  Applications, 19th International Conference, {RTA} 2008, Hagenberg, Austria,
  July 15-17, 2008, Proceedings}, Lecture Notes in Computer Science, pages
  336--350. Springer, 2008.
\newblock \href {https://doi.org/10.1007/978-3-540-70590-1_23}
  {\path{doi:10.1007/978-3-540-70590-1_23}}.

\bibitem{DBLP:conf/rta/SternagelT11}
Christian Sternagel and Ren{\'{e}} Thiemann.
\newblock Modular and certified semantic labeling and unlabeling.
\newblock In Manfred Schmidt{-}Schau{\ss}, editor, {\em Rewriting Techniques
  and Applications, 22nd International Conference, {RTA} 2011, Novi Sad,
  Serbia, May 30 - June 1, 2011}, Proceedings, LIPIcs, pages 329--344. Schloss
  Dagstuhl - Leibniz-Zentrum f{\"{u}}r Informatik, 2011.
\newblock \href {https://doi.org/10.4230/LIPICS.RTA.2011.329}
  {\path{doi:10.4230/LIPICS.RTA.2011.329}}.

\bibitem{Terese03}
Terese.
\newblock {\em Term rewriting systems}, volume~55 of {\em Cambridge Tracts in
  Theoretical Computer Science}.
\newblock Cambridge University Press, 2003.

\bibitem{DBLP:journals/entcs/ThiemannM08}
Ren{\'{e}} Thiemann and Aart Middeldorp.
\newblock Innermost termination of rewrite systems by labeling.
\newblock In J{\"{u}}rgen Giesl, editor, {\em 7th International Workshop on
  Reduction Strategies in Rewriting and Programming, WRS@RDP 2007, Paris,
  France, June 25, 2007}, Proceedings, Electronic Notes in Theoretical Computer
  Science, pages 3--19. Elsevier, 2007.
\newblock \href {https://doi.org/10.1016/J.ENTCS.2008.03.050}
  {\path{doi:10.1016/J.ENTCS.2008.03.050}}.

\bibitem{DBLP:conf/rta/Waldmann04}
Johannes Waldmann.
\newblock Matchbox: {A} tool for match-bounded string rewriting.
\newblock In Vincent van Oostrom, editor, {\em Rewriting Techniques and
  Applications, 15th International Conference, {RTA} 2004, Aachen, Germany,
  June 3-5, 2004, Proceedings}, volume 3091 of {\em Lecture Notes in Computer
  Science}, pages 85--94. Springer, 2004.
\newblock \href {https://doi.org/10.1007/978-3-540-25979-4_6}
  {\path{doi:10.1007/978-3-540-25979-4_6}}.

\bibitem{DBLP:series/eatcs/Wechler92}
Wolfgang Wechler.
\newblock {\em Universal Algebra for Computer Scientists}, volume~25 of {\em
  {EATCS} Monographs on Theoretical Computer Science}.
\newblock Springer, 1992.
\newblock \href {https://doi.org/10.1007/978-3-642-76771-5}
  {\path{doi:10.1007/978-3-642-76771-5}}.

\bibitem{DBLP:journals/fuin/Zantema95}
Hans Zantema.
\newblock Termination of term rewriting by semantic labelling.
\newblock {\em Fundam. Informaticae}, 24(1/2):89--105, 1995.
\newblock \href {https://doi.org/10.3233/FI-1995-24124}
  {\path{doi:10.3233/FI-1995-24124}}.

\bibitem{DBLP:journals/jar/Zantema05}
Hans Zantema.
\newblock Termination of string rewriting proved automatically.
\newblock {\em J. Autom. Reason.}, 34(2):105--139, 2005.
\newblock \href {https://doi.org/10.1007/S10817-005-6545-0}
  {\path{doi:10.1007/S10817-005-6545-0}}.

\end{thebibliography}

\end{document}